\documentclass[conference]{IEEEtran}
\IEEEoverridecommandlockouts
\usepackage{cite}
\usepackage{amsmath,amssymb,amsfonts}
\usepackage{graphicx}
\usepackage{textcomp}
\usepackage{balance}
\usepackage{xcolor}
\def\BibTeX{{\rm B\kern-.05em{\sc i\kern-.025em b}\kern-.08em
    T\kern-.1667em\lower.7ex\hbox{E}\kern-.125emX}}

\usepackage[caption=false,font=footnotesize]{subfig}
\usepackage{amsmath,amssymb,bm}
\usepackage{graphicx}
\usepackage{algpseudocode}
\newtheorem{proposition}{Proposition}
\newcommand{\E}{\mathbb{E}}
\newcommand{\Prb}{\Pr}
\newcommand{\Uset}{\mathcal{U}}
\newcommand{\Cset}{\mathcal{C}}
\newcounter{algcounter}
\newenvironment{ieeealgorithm}[1]{%
  \refstepcounter{algcounter}%
  \begin{figure}[t]
  \centering
  \begin{minipage}{0.98\columnwidth}
  \hrule\vspace{3pt}
  \textbf{Algorithm~\thealgcounter: #1}\par\vspace{2pt}
  \small
}{%
  \vspace{2pt}\hrule
  \end{minipage}
  \end{figure}
}

\newcommand{\safeincludegraphics}[2]{%
  \IfFileExists{#1}{%
    \includegraphics[width=\columnwidth]{#1}%
  }{%
    \fbox{\parbox[c][1.25in][c]{0.94\columnwidth}{\centering #2}}%
  }%
}

\begin{document}

%\title{Conference Paper Title*\\

\title{\LARGE QoS-Aware Joint Concurrency and HARQ-Limit Design for HARQ-CC-Aided Slow Fluid Antenna Multiple Access
%\title{Optimal Concurrency and HARQ-Limit Design for HARQ-CC-Aided Slow Fluid Antenna Multiple Access}
%\title{Joint Concurrency and HARQ Adaptation for Slow Fluid Antenna Multiple Access}

\thanks{The work of K. K. Wong is supported by the Engineering and Physical Sciences Research Council (EPSRC) under Grant EP/W026813/1.}
% \thanks{S. Han, K. K. Wong, and H. Hong are with the Department of Electronic and Electrical Engineering, University College London, London, United Kingdom. K. K. Wong is also affiliated with the Department of Electronic Engineering, Kyung Hee University, Yongin-si, Gyeonggi-do 17104, Korea (e-mail: \{sixu.han.22, kai-kit.wong, hanjiang.hong\}@ucl.ac.uk).}
% \thanks{H. Liang is with the School of Information Science and Engineering, Southeast University, Nanjing, Jiangsu Province 210096, China.}
\thanks{K. K. Wong is also affiliated with the Department of Electronic Engineering, Kyung Hee University, Yongin-si, Gyeonggi-do 17104, Korea.}
\thanks{Corresponding author: Kai-Kit Wong.}
}

\author{\IEEEauthorblockN{Sixu Han}
\IEEEauthorblockA{
\textit{University College London}\\
London, United Kingdom \\
sixu.han.22@ucl.ac.uk}
\and
\IEEEauthorblockN{Kai-Kit Wong}
\IEEEauthorblockA{
\textit{University College London}\\
London, United Kingdom \\
kai-kit.wong@ucl.ac.uk}
\and
\IEEEauthorblockN{Hanjiang Hong}
\IEEEauthorblockA{
\textit{University College London}\\
London, United Kingdom \\
hanjiang.hong@ucl.ac.uk}
\and
\IEEEauthorblockN{Haoyu Liang}
\IEEEauthorblockA{
	\textit{Southeast University}\\
	Nanjing, China \\
	lianghaoyu@seu.edu.cn}
}

\maketitle
\begin{abstract}
	Hybrid automatic repeat request with Chase combining (HARQ-CC) improves the reliability of slow fluid antenna multiple access (\emph{s}FAMA), but its transmission limit must be jointly optimized with user concurrency. This paper proposes a quality-of-service (QoS)-aware joint design of concurrency level $U$ and HARQ limit $C$ for downlink HARQ-CC-aided \emph{s}FAMA. Based on the selected-port signal-to-interference ratio (SIR) distribution, we derive multi-round decoding failure probabilities and develop a two-stage optimization that maximizes feasible concurrency and then selects the minimum $C$ satisfying outage and mean-service constraints. We prove that the first feasible $C$ minimizes service duration for a given $U$, and that feasible $U$ values form an initial segment, enabling optimal bisection search. Numerical results agree with Monte Carlo simulations under full Jakes spatial covariance. The proposed adaptive FAS design supports six concurrent users, compared with three to five for fixed-HARQ FAS and two for HARQ-adaptive fixed-position-antenna (FPA), demonstrating that joint concurrency and HARQ adaptation expands the QoS-feasible region of FAS.
\end{abstract}

\begin{IEEEkeywords}
	Fluid antenna system (FAS), hybrid automatic repeat request (HARQ), slow fluid antenna multiple access (\emph{s}FAMA), concurrency control, reliability--latency tradeoff, quality-of-service (QoS).
\end{IEEEkeywords}

\section{Introduction}
\label{sec:introduction}

\IEEEPARstart{F}{luid} antenna systems (FASs) exploit closely spaced candidate ports within a compact aperture and activate the port with the most favorable propagation condition \cite{wong2020limits,wong2021fas,FBL_FAS,new2024tutorial,dual_fas,jsac}. By exploiting port-dependent spatial fluctuations, fluid antenna multiple access (FAMA) enables multiple downlink streams to share the same time--frequency resource while mitigating multiuser interference \cite{wong2022fama}. According to the port-update time scale, FAMA is categorized into fast FAMA (\emph{f}FAMA) \cite{wong2023ffama,waqar2026fast} and slow FAMA (\emph{s}FAMA) \cite{wong2023sfama,zhang2025fundamental,hong2025multi}. Specifically, \emph{f}FAMA performs symbol-level port switching to exploit instantaneous interference variations, whereas \emph{s}FAMA selects the port maximizing the average signal-to-interference ratio (SIR) and updates it over a slower time scale.

However, one-shot \emph{s}FAMA may suffer from insufficient reliability under high user concurrency or stringent spectral-efficiency requirements. Hybrid automatic repeat request with Chase combining (HARQ-CC) improves reliability by retransmitting failed packets and combining received replicas across rounds, thereby accumulating per-round SIRs \cite{ahmed2021harq,HARQCC}. The HARQ limit is therefore coupled with the concurrency level. A larger number of active users improves spatial reuse but increases interference, whereas a larger HARQ limit enhances reliability at the cost of longer service duration. Independent selection of these two parameters may lead to excessive retransmissions or infeasible concurrency under quality-of-service (QoS) constraints.

Joint design of HARQ and resource allocation is important for increasing the supported traffic load under reliability and latency constraints \cite{malak2019throughput,anand2018resource}.
Existing HARQ-\emph{s}FAMA studies include \cite{han2026harqsfama}, which develops a joint queue--HARQ model to capture inter-user coupling induced by queueing and HARQ retransmissions, and \cite{FCN_reference}, which investigates dense-port performance under exogenous Bernoulli interferer activity.
% Existing HARQ-\emph{s}FAMA studies mainly characterize performance under stochastic user activity \cite{han2026harqsfama,FCN_reference}. In \cite{han2026harqsfama}, a joint queue--HARQ model captures the coupling among users arising from the interaction of queueing, HARQ retransmissions, and interference, whereas \cite{FCN_reference} focuses on dense-port performance analysis under exogenous Bernoulli interferer activity. 
These models do not directly address the scheduler's problem of maximizing admitted users and determining the minimum HARQ limit required for QoS satisfaction. In this work, we consider a full-buffer scheduled system with $U$ concurrently active users. Each departing packet is replaced by a new packet, so that every user experiences $U-1$ active interferers in each transmission round.
The main contributions are summarized as follows.
\begin{itemize}
	\item We jointly optimize the number of concurrently admitted users and the HARQ limit under packet-outage and mean-service-duration constraints. The required QoS metrics are characterized using the per-round selected-port SIR distribution and Stieltjes convolution.   	
	\item We establish structural properties of the joint design problem. For a fixed concurrency level, the first HARQ limit satisfying the reliability constraint also minimizes the mean service duration among all reliability feasible limits. Moreover, feasible concurrency levels form an initial segment, which removes the HARQ limit from the outer search and enables a globally optimal bisection search over concurrency.
	\item We validate the proposed design through Monte Carlo (MC) simulations with full Jakes spatial covariance and compare it with fixed-HARQ FAS and HARQ-adaptive fixed-position-antenna (FPA) benchmarks. The results show that joint concurrency and HARQ adaptation effectively extends the QoS-feasible concurrency region by translating the reliability gain of FAS into higher user admission capability.
\end{itemize}

\section{System Model}
\label{sec:system}
\subsection{System Configuration}
Consider a downlink multiuser system in which a base station (BS)
equipped with $N_t$ conventional transmit antennas simultaneously
serves $U\in\mathcal{U}\triangleq\{2,\ldots,U_{\max}\}$ user
equipments (UEs) over the same time--frequency resource, where
$N_t=U$. 
% Consider a downlink multiuser system in which a base station (BS) equipped with $N_t\geq U_{\max}$ conventional transmit antennas simultaneously serves $U\in\mathcal{U}\triangleq\{2,\ldots,U_{\max}\}$ user equipments (UEs) over the same time--frequency resource. For a given concurrency level $U$, the BS activates $U$ transmit antennas, each carrying one data stream intended for a different UE.
The transmit antennas are sufficiently separated such that
the channel vectors associated with different transmit antennas are
mutually independent. Each transmit antenna carries one data stream
intended for a different UE. Each UE employs a one-dimensional FAS
(1D-FAS) with $K$ uniformly spaced ports over an aperture of length
$W\lambda_c$, where $\lambda_c$ denotes the carrier wavelength. Each
1D-FAS is connected to a single RF chain and accesses one port at a
time.

The received signal at UE~$u$ in HARQ round $i$ is
\begin{equation}
\label{eq:ruqi_rx}
\mathbf{r}_{u}^{(i)}
=
h_{u,u,k_u^{(i)}}^{(i)}\mathbf{s}_{u}^{(i)}
+
\sum_{\substack{v=1\\v\neq u}}^{U}
h_{v,u,k_u^{(i)}}^{(i)}\mathbf{s}_{v}^{(i)}
+
\boldsymbol{\eta}_{u,k_u^{(i)}}^{(i)}.
\end{equation}
Here, $k_u^{(i)}\in\{1,\ldots,K\}$ denotes the index of the active port
of UE~$u$ in round $i$, and $\mathbf{s}_{v}^{(i)}\in\mathbb{C}^{L}$ denotes the length-$L$
symbol sequence transmitted to UE~$v$. Under
equal-power transmission, all streams have the same average symbol
energy
$E_s\triangleq L^{-1}\mathbb{E}\!\left[
\|\mathbf{s}_{v}^{(i)}\|^{2}\right]$.
Moreover, $h_{v,u,k}^{(i)}\in\mathbb{C}$ denotes the channel
coefficient from the transmit antenna carrying the stream intended for
UE~$v$ to port $k$ of UE~$u$, and
$\mathbf{h}_{v,u}^{(i)}
\triangleq
[h_{v,u,1}^{(i)},\ldots,h_{v,u,K}^{(i)}]^{\mathsf T}
\in\mathbb{C}^{K}$
denotes the corresponding channel vector across the $K$ ports.
Finally, $\boldsymbol{\eta}_{u,k}^{(i)}\in\mathbb{C}^{L}$ denotes the
AWGN vector at port $k$ of UE~$u$. 

Under the Jakes spatial correlation model~\cite{jakesFAS}, $
\mathbf{h}_{v,u}^{(i)}
\sim
\mathcal{CN}\!\left(
\mathbf{0},
\sigma^{2}\boldsymbol{\Sigma}_{\mathrm J}
\right),$
where $\sigma^{2}$ denotes the average channel power and
$\boldsymbol{\Sigma}_{\mathrm{J}}\in\mathbb{R}^{K\times K}$ is the
normalized spatial correlation matrix with entries
\begin{equation}
\label{eq:jakes}
\left[\boldsymbol{\Sigma}_{\mathrm{J}}\right]_{k,\ell}
=
J_{0}\!\left(
\frac{2\pi W(k-\ell)}{K-1}
\right),~
1\leq k,\ell\leq K.
\end{equation}
Here, $J_{0}(\cdot)$ denotes the zeroth-order Bessel function of the
first kind. Each HARQ round experiences an independent quasi-static fading block, while the spatial correlation across the $K$ FAS ports within each block is characterized by the Jakes spatial correlation model in \eqref{eq:jakes}.

For tractable analysis, $\boldsymbol{\Sigma}_{\mathrm J}$ is
approximated using the spatial block-correlation model
in~\cite{spatialblockmodel,finiteblocklengthFAS}. The $K$ ports are
partitioned into $B$ disjoint blocks
$\{\mathcal{K}_{b}\}_{b=1}^{B}$, where
$|\mathcal{K}_{b}|=L_b$ and $\sum_{b=1}^{B}L_b=K$. Let $b(k)$ denote
the block containing port $k$. Accordingly, $
\mathbf{h}_{v,u}^{(i)}
\sim
\mathcal{CN}\!\left(
\mathbf{0},
\sigma^{2}\widehat{\boldsymbol{\Sigma}}
\right),$
where
\begin{equation}
\label{eq:Sigma_hat_blockdiag}
\widehat{\boldsymbol{\Sigma}}
=
\operatorname{diag}
\left(
\mathbf{A}_{1},\ldots,\mathbf{A}_{B}
\right),
\end{equation}
with
\begin{equation}
\label{eq:Ab_common_form}
\mathbf{A}_{b}
=
(1-\mu^{2})\mathbf{I}_{L_b}
+
\mu^{2}\mathbf{1}_{L_b}\mathbf{1}_{L_b}^{\mathsf T},
~
0<\mu<1.
\end{equation}
Thus, any two distinct ports within the same block have correlation
coefficient $\mu^2$, whereas ports in different blocks are independent. The block parameters $B$, $\{L_b\}_{b=1}^{B}$, and $\mu$
are obtained from the Jakes correlation model following
\cite[Algorithm~1]{spatialblockmodel}.
Equivalently, the covariance structure in
\eqref{eq:Sigma_hat_blockdiag}--\eqref{eq:Ab_common_form} admits the
common-component representation
\begin{equation}
\label{eq:common_gauss_multiblock}
\begin{aligned}
h_{v,u,k}^{(i)}
={}&
\sigma\left(
\sqrt{1-\mu^{2}}\,x_{v,u,k}^{(i)}
+
\mu\,\bar{x}_{v,u,b(k)}^{(i)}
\right)
\\
&+
\mathrm{j}\sigma\left(
\sqrt{1-\mu^{2}}\,y_{v,u,k}^{(i)}
+
\mu\,\bar{y}_{v,u,b(k)}^{(i)}
\right),
\end{aligned}
\end{equation}
where
$x_{v,u,k}^{(i)}$, $y_{v,u,k}^{(i)}$, $\bar{x}_{v,u,b}^{(i)}$, and $\bar{y}_{v,u,b}^{(i)}$ are mutually independent $\mathcal{N}(0,1/2)$ random variables. The block-common variables $\bar{x}_{v,u,b}^{(i)}$ and $\bar{y}_{v,u,b}^{(i)}$ are shared by all ports in $\mathcal{K}_{b}$.

For a tagged UE~$u$, the noise term in~\eqref{eq:ruqi_rx} is neglected
under the interference-limited assumption. The instantaneous SIR at
port $k$ in HARQ round $i$ is therefore
\begin{equation}
\label{eq:port_sir}
\mathrm{SIR}_{U,k}^{(i)}
=
\frac{
\left|h_{u,u,k}^{(i)}\right|^{2}
}{
\displaystyle
\sum_{\substack{v=1\\v\neq u}}^{U}
\left|h_{v,u,k}^{(i)}\right|^{2}
},~
1\leq k\leq K.
\end{equation}
Under \emph{s}FAMA, UE~$u$ selects an SIR-maximizing port in each HARQ round
according to
\begin{equation}
\label{eq:port_selection}
k_u^{(i)}
\in
\operatorname*{arg\,max}_{1\leq k\leq K}
\mathrm{SIR}_{U,k}^{(i)},
\end{equation}
yielding the selected-port SIR
\begin{equation}
\label{eq:selected_sir}
\mathrm{SIR}_U^{(i)}
\triangleq
\mathrm{SIR}_{U,k_u^{(i)}}^{(i)}
=
\max_{1\leq k\leq K}
\mathrm{SIR}_{U,k}^{(i)}.
\end{equation}

For a fixed $U$, the per-round SIRs
$\{\mathrm{SIR}_{U}^{(i)}\}$ are independent and identically
distributed across HARQ rounds. Accordingly, for
$i=1,\ldots,C$, their common CDF under the adopted spatial
block-correlation model is
\begin{equation}
\label{eq:per_round_cdf}
F_U(x)
\triangleq
\Pr\!\left\{
\mathrm{SIR}_{U}^{(i)}\leq x
\right\}
=
\left.
F_{\mathrm{SIR}}(x\mid m)
\right|_{m=U-1},
\end{equation}
where $F_{\mathrm{SIR}}(x\mid m)$ is evaluated using
\cite[Eq.~(37)]{han2026harqsfama}, with $m=U-1$ corresponding to the
$U-1$ interfering streams. 
The dependence of $F_U$ on $K$ and the other fixed parameters is suppressed for notational simplicity.

\subsection{HARQ-CC Reliability and Service Metrics}
With independently randomized interfering symbol blocks, the
aggregate interference terms in distinct HARQ rounds are mutually
uncorrelated. We further assume that, for any fixed concurrency level
$U$, the per-round post-selection SIRs
$\{\mathrm{SIR}_{U}^{(j)}\}_{j=1}^{C}$ are independent and
identically distributed nonnegative random variables with common CDF
$F_U$. Chase combining of the selected-port
observations yields the accumulated SIR after $i$ transmissions as
\begin{equation}
\label{eq:combined_sir}
\Gamma_U(i)
=
\sum_{j=1}^{i}
\mathrm{SIR}_{U}^{(j)}.
\end{equation}

We assume ideal outage-based decoding with perfect error detection and
instantaneous, error-free ACK/NACK feedback. For a fixed initial
spectral efficiency $R_0$, decoding after round $i$ is successful if
$\log_2[1+\Gamma_U(i)]\geq R_0$, or equivalently,
$\Gamma_U(i)\geq\theta$, where
$\theta\triangleq 2^{R_0}-1$.
Accordingly, for $i=1,\ldots,C$, the decoding-failure probability after
$i$ transmissions is
\begin{equation}
\label{eq:qi_full}
q_i(U;K,R_0)
\triangleq
\Prb\!\left\{\Gamma_U(i)<\theta\right\},
~
q_0\triangleq 1.
\end{equation}
Since $K$ and $R_0$ are fixed within each optimization instance, we
henceforth write $q_i(U)$ for notational simplicity.

Let
\begin{equation}
H_{U,i}(x)
\triangleq
\Prb\!\left\{\Gamma_U(i)\leq x\right\},
~ x\geq0.
\end{equation}
Since $\Gamma_U(1)=\mathrm{SIR}_U^{(1)}$, we have
$H_{U,1}(x)=F_U(x)$. Independence across HARQ rounds gives the
Stieltjes-convolution recursion
\begin{equation}
\label{eq:convolution}
H_{U,i}(x)
=
\int_{[0,x]}
H_{U,i-1}(x-z)\,\mathrm{d}F_U(z),
~ i=2,\ldots,C.
\end{equation}
Because the induced per-round SIR distribution $F_U$ is continuous,
$H_{U,i}$ is also continuous, and thus
\begin{equation}
q_i(U)=H_{U,i}(\theta),
~ i=1,\ldots,C.
\end{equation}

Let $C\in\Cset\triangleq\{1,\ldots,C_{\max}\}$ be the maximum
number of transmissions allowed for one packet, including its initial
transmission.  In contrast, let
\begin{equation}
\begin{aligned}
\tau_U
&\triangleq
\inf\!\left\{i\geq 1:\Gamma_U(i)\geq\theta\right\},\\
N_{\rm tx}(U,C)
&\triangleq
\min\{\tau_U,C\}
\end{aligned}
\label{eq:actual_rounds}
\end{equation}
denote the first successful round and the actual random number of
occupied transmissions, respectively.  

The packet outage probability, average number of occupied transmissions, and average HARQ service duration are
\begin{align}
 P_{\rm out}(U,C)&=q_C(U),
 \label{eq:pout}\\
 \bar N_{\rm tx}(U,C)
 &=\E[N_{\rm tx}(U,C)]
 =1+\sum_{i=1}^{C-1}q_i(U),
 \label{eq:mean_rounds}\\
 D_{\rm svc}(U,C)&=T_F\bar N_{\rm tx}(U,C),
 \label{eq:service_delay}
\end{align}
where $T_F$ is the duration of one HARQ data round. The metric $D_{\rm svc}(U,C)$ accounts only for the HARQ service time and excludes queueing delay.

To characterize resource efficiency, we define the per-user and aggregate resource-normalized goodputs as
\begin{equation}
\begin{aligned}
G_u(U,C)
&=
\frac{R_0\left[1-P_{\rm out}(U,C)\right]}
{\bar N_{\rm tx}(U,C)},\\
G_{\Sigma}(U,C)
&=
U G_u(U,C).
\end{aligned}
\label{eq:goodput}
\end{equation}
\begin{ieeealgorithm}{Bisection-Based Two-Stage Design of $U$ and $C$}
	\label{alg:search}
	\begin{algorithmic}[1]
		\Require $U_{\max},C_{\max},K,R_0,\epsilon,D_{\max},T_F$
		\Ensure $(U^{\star},C^{\star})$, or infeasible
		\Function{Test}{$U$}
		\State Evaluate $F_U$ and
		$q_i(U)$, $i=1,\ldots,C_{\max}$
		\State $c\gets\min\{i:q_i(U)\leq\epsilon\}$
		\If{no such $c$ exists}
		\State \Return $(\mathrm{false},+\infty)$
		\EndIf
		\State $d\gets T_F[1+\sum_{i=1}^{c-1}q_i(U)]$
		\State \Return $(d\leq D_{\max},c)$
		\EndFunction
		\State $\ell\gets2$, $r\gets U_{\max}$, $b\gets0$
		\While{$\ell\leq r$}
		\State $m\gets\lfloor(\ell+r)/2\rfloor$
		\State $(\mathrm{ok},c_m)\gets\Call{Test}{m}$
		\If{$\mathrm{ok}$}
		\State $(b,c_b)\gets(m,c_m)$; $\ell\gets m+1$
		\Else
		\State $r\gets m-1$
		\EndIf
		\EndWhile
		\State \Return infeasible if $b=0$; otherwise return $(b,c_b)$
	\end{algorithmic}
\end{ieeealgorithm}

\section{Optimization Variables\\ and Problem Formulation}
\label{sec:optimization}
The variable $U$ controls scheduled spatial reuse and, simultaneously,
the deterministic interferer count $U-1$.  The variable $C$ controls the maximum HARQ protection available to each packet.  By contrast, $K$ is a receiver-hardware parameter and $R_0$ is fixed by the selected payload and modulation-and-coding mode; neither is adapted by the optimization. They may be swept externally to assess sensitivity, but every reported optimizer must correspond to one fixed $(K,R_0)$ pair.

Let $\epsilon$ and $D_{\max}$ denote the target packet outage
probability and maximum average service duration, respectively. For
convenience, define the normalized mean-service budget as
\begin{equation}
d_{\max}\triangleq\frac{D_{\max}}{T_F}.
\label{eq:dmax}
\end{equation}
We seek the largest number of users that can be served concurrently.
Among policies supporting that concurrency, we choose the smallest HARQ
limit. The resulting two-stage optimization problem is formulated as
\begin{subequations}
\label{prob:joint_design}
\begin{align}
U^{\star}
&=
\max_{U\in\Uset,\,C\in\Cset} U
\label{prob:primary}\\
\mathrm{s.t.}\quad
&q_C(U)\leq\epsilon,
\label{prob:reliability}\\
&\bar N_{\rm tx}(U,C)\leq d_{\max},
\label{prob:latency}
\end{align}
\end{subequations}
with the tie-breaking rule
\begin{equation}
 C^{\star}=\min\{C\in\Cset:(U^{\star},C)
 \text{ satisfies }\eqref{prob:reliability}\text{--}\eqref{prob:latency}\}.
 \label{eq:tie_break}
\end{equation}
The resulting design therefore characterizes the maximum admissible
concurrency subject to per-user reliability and latency constraints.
The aggregate resource-normalized goodput $G_{\Sigma}$ in
\eqref{eq:goodput} is used only as a supplementary performance metric
and does not replace these QoS requirements.

%\section{Structural Reduction and Bisection Search}

\section{Structural Reduction of the Two-Stage Design}
\label{sec:algorithm}
For a candidate concurrency $U$, define the first HARQ limit that meets
the reliability target as
\begin{equation}
 C_{\epsilon}(U)
 =\min\{c\in\Cset:q_c(U)\leq\epsilon\},
 \label{eq:cepsilon}
\end{equation}
where $C_{\epsilon}(U)=+\infty$ if the set is empty.

\begin{proposition}\label{prop:structure}
For fixed $(K,R_0)$, a concurrency level $U\in\Uset$ is feasible under
the reliability and service-latency constraints if and only if
$C_{\epsilon}(U)<+\infty$ and
\begin{equation}
T_F\!\left[
1+\sum_{i=1}^{C_{\epsilon}(U)-1}q_i(U)
\right]
\leq D_{\max}.
\label{eq:oracle}
\end{equation}
Moreover, if $U_2\in\Uset$ is feasible, then every
$U_1\in\Uset$ satisfying $U_1<U_2$ is also feasible. 
Hence, the feasible concurrency levels are contiguous from the smallest admissible value up to the maximum feasible concurrency.
\end{proposition}

\begin{IEEEproof}
Since every selected SIR is nonnegative,
$\Gamma_U(i+1)\geq\Gamma_U(i)$, and hence
$q_{i+1}(U)\leq q_i(U)$. Therefore, any HARQ limit satisfying the
reliability constraint must satisfy $C\geq C_{\epsilon}(U)$. Moreover,
\begin{equation}
D_{\rm svc}(U,C+1)-D_{\rm svc}(U,C)
=T_F q_C(U)\geq0,
\label{eq:delay_increment}
\end{equation}
so $D_{\rm svc}(U,C)$ is nondecreasing in $C$. Hence, if
$C_{\epsilon}(U)<+\infty$, the choice $C=C_{\epsilon}(U)$ yields the
minimum mean service duration among all HARQ limits satisfying the
reliability constraint. It follows that $U$ is feasible if and only if
$C_{\epsilon}(U)<+\infty$ and
$D_{\rm svc}(U,C_{\epsilon}(U))\leq D_{\max}$, proving the first
statement.

For $U_2>U_1$, couple the two systems using the same desired channels
and the same first $U_1-1$ interfering channels.  Adding the remaining
$U_2-U_1$ nonnegative interference powers decreases every port-wise SIR
almost surely.  Maximization over the common $K$ ports and summation over
rounds preserve this ordering, giving
$q_i(U_2)\geq q_i(U_1)$ for all $i$.  It follows that
$C_{\epsilon}(U_1)\leq C_{\epsilon}(U_2)$ and
\begin{align*}
 D_{\rm svc}(U_1,C_{\epsilon}(U_1))
 &\leq T_F\!\left[1+\sum_{i=1}^{C_{\epsilon}(U_1)-1}q_i(U_2)\right]\\
 &\leq D_{\rm svc}(U_2,C_{\epsilon}(U_2)).
\end{align*}
Thus, feasibility of $U_2$ implies feasibility of $U_1$.
\end{IEEEproof}

Proposition~\ref{prop:structure} eliminates $C$ from the outer search and enables bisection over $U$. Therefore, once the largest feasible concurrency $U^\star$ is found,
the second-stage solution is simply
$C^\star=C_{\epsilon}(U^\star)$.
The adopted spatial block correlation model preserves the monotonic interference structure of the original
Jakes spatial correlation model: increasing $U$ adds nonnegative
interference power at every port while leaving the desired channel
unchanged. Hence, the ordering of ${q_i(U)}$ and the resulting
feasibility monotonicity remain valid under the block correlation
model. Consequently, Algorithm~\ref{alg:search} solves the two-stage design
by bisection over $U$ and returns the optimal pair
$(U^\star,C^\star)$. Bisection reduces the number of feasibility evaluations from $O(U_{\max})$ to $O(\log U_{\max})$, providing greater computational benefits for larger search spaces.

%\begin{ieeealgorithm}{Joint Search for $(U^{\star},C^{\star})$}
%\begin{ieeealgorithm}{Bisection Search for the Optimal Concurrency}
\begin{figure}[t!]
	\centering
	\includegraphics[width=0.79\columnwidth]{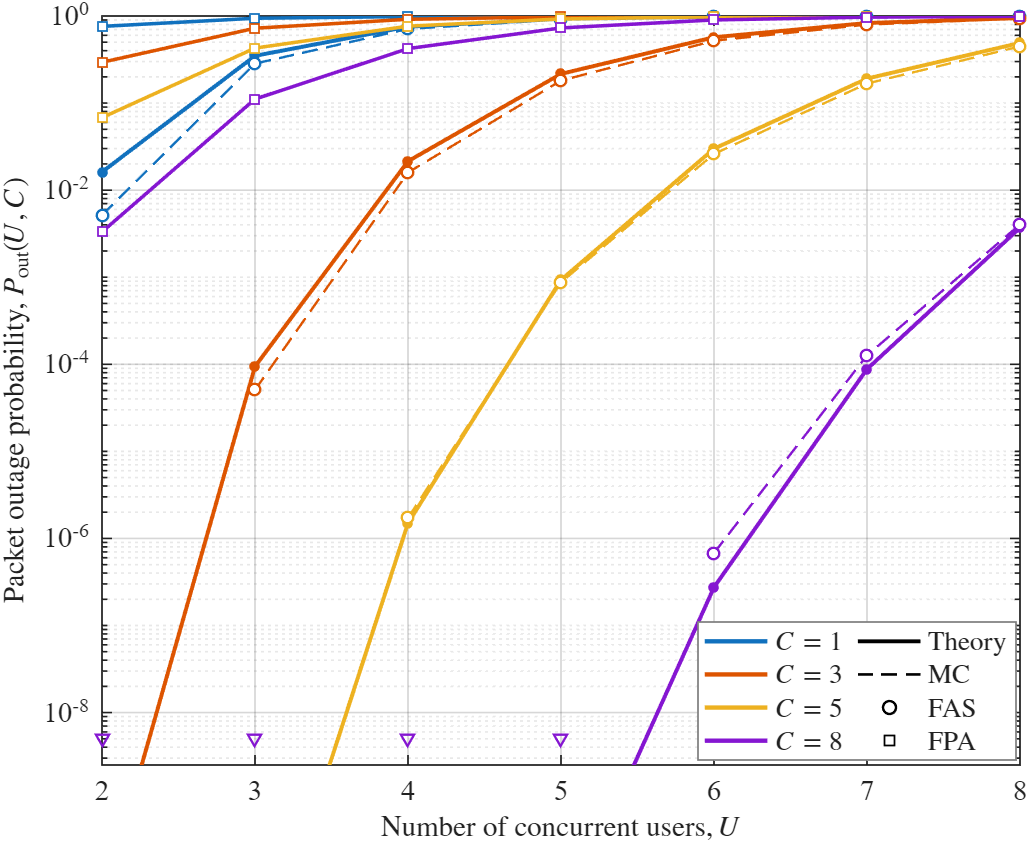}
	\caption{Packet-outage probability $P_{\mathrm{out}}(U,C)$ versus
		concurrency level $U$. Solid and dashed curves denote analytical and
		MC results, respectively, while circles and squares indicate FAS and
		FPA. Downward triangles represent zero observed outages placed at
		$1/N_{\rm MC}$ for visibility.}
	\label{fig:pout_vs_U}
	\vspace{-2mm}
\end{figure}

\section{Numerical Results}
\label{sec:numerical}
Unless otherwise stated, we set the carrier frequency to $7$ GHz and
the FAS aperture to $D=0.15$ m, yielding $W=D/\lambda_c=3.5024$ with
$K=48$ uniformly spaced ports. We consider $U\in\{2,\ldots,8\}$ active
users with $m=U-1$ interferers, $\theta=5$ dB, $T_F=1$ ms, and
$C_{\max}=12$. The QoS parameters are varied over
$\epsilon\in[10^{-4},10^{-1}]$ and
$d_{\max}=D_{\max}/T_F\in[1,5]$.
The analytical results use a $30\times30$ Gauss--Laguerre quadrature and
a $6000$-bin FFT convolution, while MC benchmarks use
$N_{\rm MC}=2\times10^{8}$ packet realizations with the full Jakes
spatial covariance. An FPA receiver is adopted as the conventional
benchmark.

Fig.~\ref{fig:pout_vs_U} validates the outage analysis. The analytical
results closely match the full-Jakes MC simulations over the resolvable
probability range. For fixed $C$, the outage probability increases with
$U$ due to stronger interference, whereas larger $C$ improves reliability
through HARQ-CC accumulation. FAS consistently outperforms FPA under
identical $(U,C)$ configurations. Fig.~\ref{fig:reliability_latency_tradeoff} illustrates the tradeoff between reliability and service duration. Increasing $C$ reduces outage probability but increases transmission occupancy, while larger $U$ causes higher interference and service demand. Under $\epsilon=10^{-3}$ and $d_{\max}=4$, the proposed two-stage design selects the maximum feasible concurrency and then the minimum HARQ limit, yielding $(U^\star,C^\star)=(6,7)$.

\begin{figure}[t!]
	\centering
	\includegraphics[width=0.79\columnwidth]{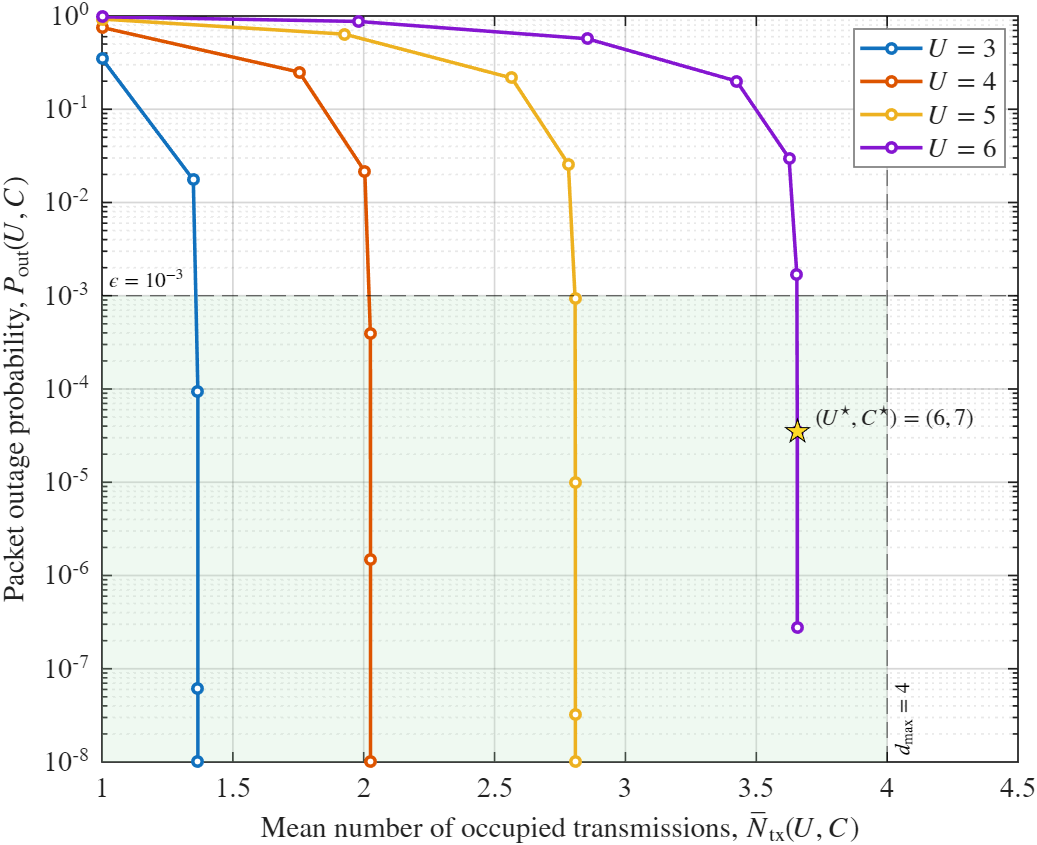}
	\caption{Reliability--service tradeoff as the HARQ limit $C$ varies for
		different concurrency levels $U$. The feasible region satisfies both
		outage and mean-service constraints at $\epsilon=10^{-3}$ and
		$d_{\max}=4$.}
	\label{fig:reliability_latency_tradeoff}
	\vspace{-2mm}
\end{figure}

\begin{figure*}[t!]
	\centering
	
	\subfloat[$U^\star$ versus $\epsilon$ at $d_{\max}=4$]
	{\includegraphics[width=.303\linewidth]{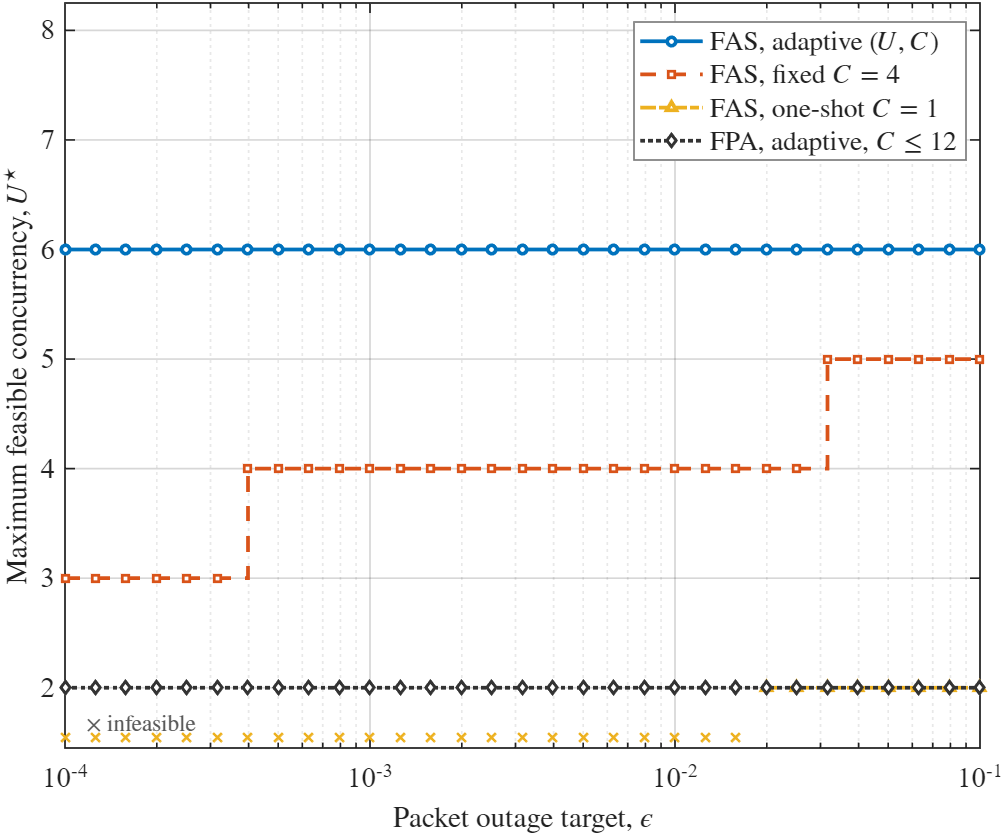}
		\label{fig:Ustar_vs_epsilon}}
	\hspace{1mm}
	\subfloat[Maximum feasible concurrency $U^\star$]
	{\includegraphics[width=.32\linewidth]{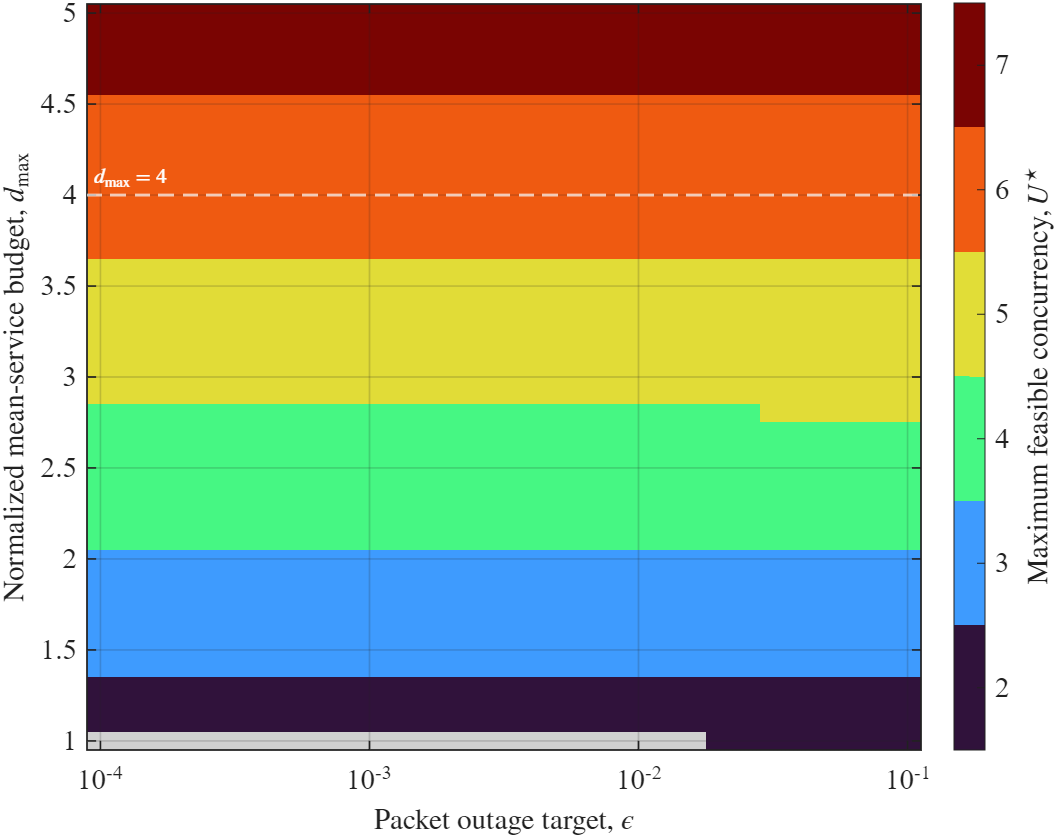}
		\label{fig:optimal_design_maps_U}}
	\hspace{1mm}
	\subfloat[Minimum HARQ limit attaining $U^\star$]
	{\includegraphics[width=.32\linewidth]{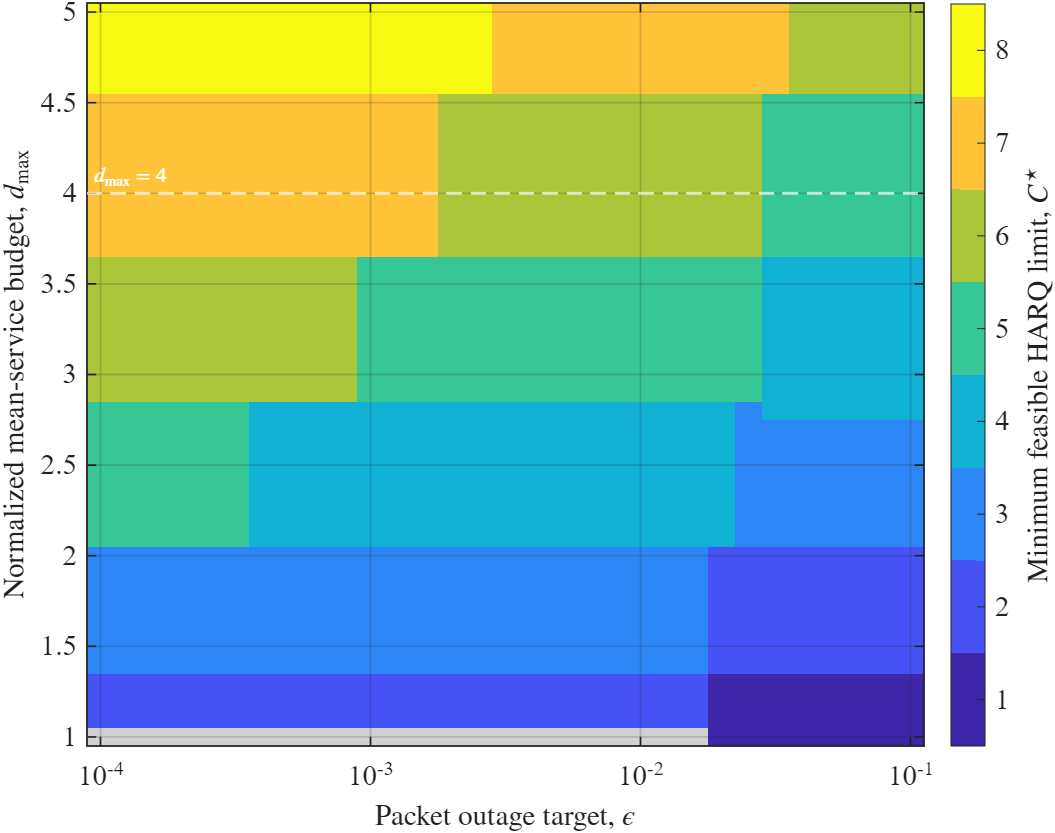}
		\label{fig:optimal_design_maps_C}}
	
	\caption{QoS-aware design results at $\theta=5$ dB and $C_{\max}=12$.
		Gray regions indicate infeasible QoS pairs.}
	\label{fig:optimal_design_results}
	\vspace{-2mm}
\end{figure*}

Fig.~\ref{fig:optimal_design_results} summarizes the system-level
optimization. Fig.~\ref{fig:Ustar_vs_epsilon} compares the proposed
adaptive design with fixed-HARQ FAS and adaptive FPA at $d_{\max}=4$.
The adaptive-FAS scheme maintains $U^\star=6$ over the considered outage
range, whereas fixed-HARQ FAS and FPA achieve smaller concurrency levels.
Figs.~\ref{fig:optimal_design_maps_U} and
\ref{fig:optimal_design_maps_C} extend the evaluation to the complete
$(\epsilon,d_{\max})$ design space. The feasible concurrency $U^\star$
increases monotonically with relaxed outage and service constraints and
exhibits a staircase structure due to integer-valued user admission. The
required HARQ limit $C^\star$ generally decreases with relaxed outage
requirements, but may increase with larger $d_{\max}$ because additional
service budget enables higher concurrency with stronger interference.
These results confirm the effectiveness of the proposed two-stage
optimization.

\section{Conclusion}
\label{sec:conclusion}
This paper jointly optimized concurrency and HARQ limits for downlink HARQ-CC-aided \emph{s}FAMA under outage and mean-service constraints. Selected-port SIR distributions yielded multi-round failure probabilities and QoS metrics via Stieltjes convolution. Monotonicity enabled a two-stage global design: maximize feasible concurrency via bisection, then select the minimum feasible HARQ limit. Full-Jakes simulations validated the analysis and showed that adaptation converts FAS reliability gains into higher QoS-constrained concurrency.

\balance


\begin{thebibliography}{99}

\bibitem{wong2020limits}
K.-K. Wong, A. Shojaeifard, K.-F. Tong, and Y. Zhang,
``Performance limits of fluid antenna systems,''
\emph{IEEE Commun. Lett.}, vol.~24, no.~11, pp.~2469--2472,
Nov. 2020.

\bibitem{wong2021fas}
K.-K. Wong, A. Shojaeifard, K.-F. Tong, and Y. Zhang,
``Fluid antenna systems,''
\emph{IEEE Trans. Wireless Commun.}, vol.~20, no.~3,
pp.~1950--1962, Mar. 2021.
\bibitem{FBL_FAS}
Z. Zhang, K.-K. Wong, D. Morales-Jimenez, H. Jiang, H. Xu, C. Masouros, and Z. Zhang, ``Finite-blocklength fluid antenna systems,'' \emph{IEEE Trans. Wireless Commun.}, early access, doi: 10.1109/TWC.2026.3723456, 2026.
\bibitem{new2024tutorial}
W. K. New \emph{et al.},
``A tutorial on fluid antenna system for 6G networks: Encompassing communication theory, optimization methods and hardware designs,''
\emph{IEEE Commun. Surveys Tuts.}, vol.~27, no.~4, pp.~2325--2377,
2024.
\bibitem{dual_fas}
Z. Zhang, Y. Wu, K.-K. Wong, H. Jiang, and A. Li, ``Jointly correlated dual-side fluid antenna system,'' \emph{IEEE Wireless Commun. Lett.}, vol. 15, pp. 4370--4374, 2026.
% \bibitem{faa}
% Z. Zhang, K.-K. Wong, H. Jiang, F. Rostami Ghadi, H. Shin, and Y. Zhang, ``Finite-aperture fluid antenna array design: Analysis and algorithm,'' \emph{IEEE Wireless Commun. Lett.}, vol. 15, pp. 3199--3203, 2026.
\bibitem{jsac}
Z. Zhang, K.-K. Wong, J. Dang, Z. Zhang, and C.-B. Chae, ``On fundamental limits for fluid antenna-assisted integrated sensing and communications for unsourced random access,'' {\em IEEE J. Sel. Areas Commun.}, vol. 44, pp. 136--149, 2026.
\bibitem{wong2022fama}
K.-K. Wong and K.-F. Tong,
``Fluid antenna multiple access,''
\emph{IEEE Trans. Wireless Commun.}, vol.~21, no.~7,
pp.~4801--4815, Jul. 2022.
\bibitem{wong2023ffama}
K.-K. Wong, K.-F. Tong, Y. Chen, and Y. Zhang,
``Fast fluid antenna multiple access enabling massive connectivity,''
\emph{IEEE Commun. Lett.}, vol.~27, no.~2, pp.~711--715,
Feb. 2023.

\bibitem{waqar2026fast}
N. Waqar, K.-K. Wong, C.-B. Chae, and R. Murch,
``Fast fluid antenna multiple access,''
\emph{arXiv preprint arXiv:2605.23642}, 2026.

\bibitem{wong2023sfama}
K.-K. Wong, D. Morales-Jimenez, K.-F. Tong, and C.-B. Chae,
``Slow fluid antenna multiple access,''
\emph{IEEE Trans. Commun.}, vol.~71, no.~5, pp.~2831--2846,
May 2023.

\bibitem{zhang2025fundamental}
Z. Zhang, K.-K. Wong, J. Dang, Z. Zhang, C. Masouros, and C.-B. Chae,
``On fundamental limits of slow-fluid antenna multiple access for unsourced random access,''
\emph{IEEE Wireless Commun. Lett.}, vol.~14, no.~11, pp.~3455--3459,
Nov. 2025.

\bibitem{hong2025multi}
H. Hong, K.-K. Wong, X. Zhu, H. Xu, H. Xiao, F. R. Ghadi, and H. Shin,
``Multi-port selection for FAMA: Massive connectivity with fewer RF chains than users,''
\emph{arXiv preprint arXiv:2511.17897}, 2025.



\bibitem{HARQCC}
H. A. Ngo and L. Hanzo, ``Hybrid automatic-repeat-request systems for cooperative wireless communications,'' {\em IEEE Commun. Surveys Tuts.}, vol. 16, no. 1, pp. 25--45, 2014.


\bibitem{ahmed2021harq}
A. Ahmed, A. Al-Dweik, Y. Iraqi, H. Mukhtar, M. Naeem, and E. Hossain,
``Hybrid automatic repeat request (HARQ) in wireless communications
systems and standards: A contemporary survey,''
\emph{IEEE Commun. Surveys Tuts.}, vol.~23, no.~4,
pp.~2711--2752, 4th Quart. 2021.

\bibitem{malak2019throughput}
D. Malak, H. Huang, and J. G. Andrews,
``Throughput maximization for delay-sensitive random access communication,''
\emph{IEEE Trans. Wireless Commun.}, vol.~18, no.~1,
pp.~709--723, Jan. 2019.

\bibitem{anand2018resource}
A. Anand and G. de Veciana,
``Resource allocation and HARQ optimization for URLLC traffic in 5G wireless networks,''
\emph{IEEE J. Sel. Areas Commun.}, vol.~36, no.~11,
pp.~2411--2421, Nov. 2018.

\bibitem{han2026harqsfama}
S. Han, K.-K. Wong, H. Hong, and H. Shin, ``HARQ for slow fluid antenna multiple access,'' \emph{arXiv preprint arXiv:2607.18426}, 2026.

\bibitem{FCN_reference}
S. Han, K.-K. Wong, and H. Hong, ``HARQ-CC-aided slow fluid antenna multiple access with highly correlated ports: An LST-based performance analysis,'' \emph{arXiv preprint arXiv:2608.24614}, 2026.

% \bibitem{AD_Geo}
% Z. Zhang \emph{et al.}, ``Geometry-structured channel reconstruction for conventional and fluid antenna systems: Bayesian inference and fundamental limits,'' \emph{arXiv preprint}, 2026. [Online]. Available: https://arxiv.org/abs/2606.04001.
% \bibitem{ad_jstsp}
% Z. Zhang \emph{et al.}, ``Joint activity detection and channel estimation for fluid antenna system exploiting geographical and angular information,'' \emph{IEEE J. Sel. Topics Signal Process.}, vol. 20, no. 3, pp. 354--370, 2026.
% \bibitem{csi_wu}
% Y. Wu, Z. Zhang, H. Jiang, K.-K. Wong, and C.-B. Chae, ``Learned-approximate message passing under Karhunen--Lo\`eve modeling for fluid antenna systems,'' \emph{IEEE Wireless Commun. Lett.}, vol. 15, pp. 2719--2723, 2026.

\bibitem{jakesFAS}
M. Khammassi, A. Kammoun, and M.-S. Alouini, ``A new analytical
approximation of the fluid antenna system channel,'' {\em IEEE Trans.
Wireless Commun.}, vol. 22, no. 12, pp. 8843--8858, Dec. 2023.
\bibitem{spatialblockmodel}
P. Ram\'{i}rez-Espinosa, D. Morales-Jimenez, and K.-K. Wong, ``A new spatial block-correlation model for fluid antenna systems,'' {\em IEEE Trans. Wireless Commun.}, vol. 23, no. 11, pp. 15 829--15 843, 2024.
\bibitem{finiteblocklengthFAS}
Z. Zhang, K.-K. Wong, D. Morales-Jimenez, H. Jiang, P. Ram\'{i}rez-Espinosa, C.-B. Chae, and C. Masouros, ``Finite-blocklength fluid antenna systems with spatial block-correlation channel model,'' {\em IEEE Wireless Commun. Lett.}, vol. 15, pp. 1911--1915, 2026.
\end{thebibliography}
\end{document}